\documentclass[letterpaper, 10 pt, conference]{ieeeconf}

\IEEEoverridecommandlockouts                              %
\makeatletter

\let\proof\@undefined
\let\endproof\@undefined
\makeatother

 \usepackage{graphicx}          %
 \usepackage{amssymb}
 \usepackage{amsmath}
\usepackage{amsthm}

\newtheorem{thm}{Theorem}

\newtheorem{prop}{Proposition}

\theoremstyle{definition}
\usepackage{thmtools}

\theoremstyle{plain}
\declaretheorem[style=definition,qed=$\blacksquare$]{example}
\newtheorem{definition}{Definition}

\newtheorem{corollary}{Corollary}
\theoremstyle{remark}

\theoremstyle{definition}

\newcommand{\Domain}{\mathcal{X}}

\newcommand{\ul}{\underline}
\newcommand{\ol}{\overline}
\newcommand{\Kinf}{\mathcal{K}_\infty}
\newcommand*\wc{{\mkern 2mu\cdot\mkern 2mu}}

\title{Separability of Lyapunov Functions for Contractive Monotone Systems}
\author{Samuel Coogan\thanks{S. Coogan  is with the Electrical Engineering Department, University of California, Los Angeles. \texttt{scoogan@ucla.edu}.}}
\date{}

\renewcommand\footnotemark{}
\begin{document}
\maketitle

\vspace{-.3in}
\begin{abstract}
We consider constructing Lyapunov functions for systems that are both monotone and contractive with respect to a weighted one norm or infinity norm. This class of systems admits separable Lyapunov functions that are either the sum or the maximum of a collection of functions of a single argument. In either case, two classes of separable Lyapunov functions exist: the first class is separable along the system's state, and the second class is separable along components of the system's vector field. The latter case is advantageous for many practically motivated systems for which it is difficult to measure the system's state but easier to measure the system's velocity or rate of change. We provide several examples to demonstrate our results.
\end{abstract}

\section{Introduction}
A dynamical system is \emph{monotone} if it maintains a partial ordering of states along trajectories of the system \cite{Hirsch:1983lq,Hirsch:1985fk, Smith:2008fk}.
Monotone systems exhibit structure and ordered behavior that leads to techniques for analysis and control, \emph{e.g.}, \cite{Angeli:2003fv, Angeli:2004qy,Coogan:2014ty}. Examples of monotone systems include certain biological networks \cite{Sontag:2007ad} and transportation networks \cite{Gomes:2008fk, Lovisari:2014yq, coogan2015compartmental}, and monotone systems theory has been applied to large-scale analysis and distributed control \cite{Rantzer:2012fj, Dirr:2015rt}. 

A dynamical system is \emph{contractive} if the distance between states along any pair of trajectories is exponentially decreasing \cite{Pavlov:2004lr, LOHMILLER:1998bf, Sontag:2010fk,Forni:2012qe}. When an equilibrium exists, contraction implies global convergence and a Lyapunov function is given by the distance to the equilibrium. The magnitude of the vector field provides an alternative Lyapunov function.  

Certain classes of monotone systems have been shown to be also contractive with respect to non-Euclidean norms. For example, \cite{Margaliot:2012hc, Margaliot:2014qv, Raveh:2015wm} study a model for gene translation which is monotone and contractive with respect to a weighted $\ell_1$ norm. A closely related result is obtained for transportation flow networks in \cite{Coogan:2014ph, Como:2015ne}. In \cite{Coogan:2014ph}, a Lyapunov function defined as the magnitude of the vector field is used, while a Lyapunov function based on the distance of the state to the equilibrium is used in \cite{Como:2015ne}.

In this paper, we study monotone systems that are contractive with respect to a weighted $\ell_1$ norm or $\ell_\infty$ norm. We first provide sufficient conditions establishing contraction for monotone systems in terms of strict negativity of scaled row or column sums of the Jacobian matrix for the system. These conditions follow naturally from sufficient conditions for monotonicity and for contraction. 
Next, we derive sum-separable and max-separable Lyapunov functions based on the contraction metric. In particular, we introduce the notion of Lyapunov functions that are separable along components of the vector field. This is especially relevant for certain classes of systems such as multiagent control systems or flow networks where it is often more practical to measure velocity or flow rather than position or state.   Additionally, we present results of independent interest for proving asymptotic stability and obtaining Lyapunov functions of systems that are \emph{nonexpansive} with respect to a particular vector norm, \emph{i.e.}, the distance between states along any pair of trajectories does not increase. Finally, we draw connections between our results and related results, particularly small-gain theorems for interconnected, nonlinear systems.

For linear monotone systems, also called \emph{positive} systems, scalable stability verification is possible by appealing to linear Lyapunov functions or scaled componentwise-maximum Lyapunov functions \cite{Rantzer:2012fj}. We extend stability verification results from \cite{Rantzer:2012fj} to nonlinear monotone systems by appealing to contraction theory. Separable Lyapunov functions for nonlinear monotone systems are also studied in \cite{Dirr:2015rt}, however, a contraction theoretic approach is not considered. We also introduce a novel notion of flow separable Lyapunov functions not considered in \cite{Dirr:2015rt}.%

This paper is organized as follows. Section \ref{sec:notation} establishes notation and Section \ref{sec:problem-setup} provides the problem setup. Section \ref{sec:contract} reviews contraction theory and provides a novel approach for establishing asymptotic stability for nonexpansive systems. Section \ref{sec:contr-monot-syst} provides the main results, and illustrative examples are considered in Section \ref{sec:examples}. We provide discussion in Section \ref{sec:disc-comp-exist} and concluding remarks in Section \ref{sec:conclusions}.

\section{Notation}
\label{sec:notation}
A matrix $A\in\mathbb{R}^{n\times n}$ is \emph{Metzler} if all of its off diagonal components are nonnegative. All inequalities are interpreted elementwise. %
The vector of all ones is denoted by $\mathbf{1}$. For scalar functions of one variable, we denote derivative with the prime notation $'$. The $\ell_1$ and $\ell_\infty$ norms are denote by $|\cdot|_1$ and $|\cdot|_\infty$ respectively, that is, $|x|_1=\sum_{i=1}^n|x_i|$ and $|x|_\infty=\max_{i=1,\ldots,n}|x_i|$ for $x\in\mathbb{R}^n$.

\section{Problem Setup}
\label{sec:problem-setup}
We consider the nonlinear autonomous dynamical system
\begin{align}
  \label{eq:1}
  \dot{x}=f(x)
\end{align}
for $x\in\Domain\subseteq \mathbb{R}^n$ and continuously differentiable $f(\cdot)$. Let $f_i(x)$ indicate the $i$th component of $f$ and denote the Jacobian as  $J(x)= \frac{\partial f}{\partial x}(x)$.

Denote by $\phi(t,x_0)$ the solution to \eqref{eq:1} at time $t$ when the system is initialized with state $x_0$ at time $0$. We assume that \eqref{eq:1} is forward complete and $\Domain$ is forward invariant for \eqref{eq:1}, thus $\phi(t,x_0)\in \Domain$ for all $t\geq 0$ and all $x_0\in\Domain$.  In this paper, we consider $\Domain=\mathbb{R}^n$ or $\Domain=\mathbb{R}^n_{\geq 0}:=\{x\in\mathbb{R}^n\mid x\geq 0\}$.

Consider a forward invariant set $K\subset \Domain$ with $x^*\in K$ an equilibrium for \eqref{eq:1} for which the domain of attraction includes $K$.  Let  $V:K\to\mathbb{R}_{\geq 0}$ be a \emph{Lyapunov function} for \eqref{eq:1} on $K$ that establishes asymptotic stability of $x^*$, that is: $V(x)$ is continuous and $V(x)=0$ for $x\in K$ if and only if $x=x^*$; $V(x)$ is radially unbounded; and $V(\phi(t,x_0))$ is a nonincreasing function of $t$ and $\lim_{t\to\infty}V(\phi(t,x_0))=0$ for all $x_0\in K$.  In this paper, we consider non-differentiable Lyapunov functions, for which standard Lyapunov theory can be extended using generalized derivatives \cite{Dirr:2015rt}.

The Lyapunov function $V(x)$ is \emph{state sum-separable} if %
\begin{align}
  \label{eq:46}
  V(x)=\sum_{i=1}^nV_i(x_i)
\end{align}
for a collection of functions $V_i$. It is \emph{state max-separable} if
\begin{align}
  \label{eq:47}
    V(x)=\max_{i=1,\ldots,n}V_i(x_i).
\end{align}
Lyapunov functions decomposable as \eqref{eq:46} or \eqref{eq:47} are considered in \cite{Dirr:2015rt}. In this paper, we also consider Lyapunov functions that are separable based on the dynamics of \eqref{eq:1}. If %
\begin{align}
  \label{eq:48}
  V(x)&=\sum_{i=1}^nW_i(f_i(x)),\quad \text{or}\\
  V(x)&=\max_{i=1,\ldots,n}W_i(f_i(x))
\end{align}
for some collection of functions $W_i$, we say that $V$ is, respectively, \emph{flow sum-separable} and \emph{flow max-separable}.

Except in Section \ref{sec:contract}, we assume \eqref{eq:1} is monotone:
\begin{definition}
\label{def:mon}
The system \eqref{eq:1} is \emph{monotone} if the dynamics maintain a partial order on solutions, that is,
\begin{align}
  \label{eq:3}
  x_0\leq y_0 \implies \phi(t,x_0)\leq \phi(t,y_0) \quad \forall t\geq 0.
\end{align}  
\end{definition}
In this paper, monotonicity is defined with respect to the positive orthant, although it is common to consider monotonicity with respect to other cones \cite{Angeli:2003fv}.

\begin{prop}[{\cite[Ch. 3.1]{Smith:2008fk}}]
\label{prop:mono}
The system \eqref{eq:1} is monotone if and only if the Jacobian $J(x)$ is Metzler for all $x\in\Domain$.  
\end{prop}

\section{Infinitesimal Contraction}
\label{sec:contract}
We now review infinitesimal contraction for autonomous dynamical systems. %
We again consider the system given in \eqref{eq:1} but momentarily disregard the assumption that the system is monotone. Let $|\wc|$ be a vector norm on $\mathbb{R}^n$ and let $\|\wc \|$ be its induced matrix norm on $\mathbb{R}^{n\times n}$. The corresponding \emph{matrix measure} of the matrix $A\in\mathbb{R}^{n\times n}$ is defined as
\begin{align}
  \label{eq:5}
  \mu(A):= \lim_{h\to 0^+}\frac{\|I+hA\|-1}{h}.
\end{align}

\begin{prop}
\label{thm:contract}
  Let $K\subseteq \Domain$ be convex and forward-invariant. If, for some $c\in\mathbb{R}$,
  \begin{align}
    \label{eq:15}
      \mu(J(x))\leq c \quad \forall x\in K,
  \end{align}
then for any two solutions $x(t)=\phi(t,x_0)$ and $y(t)=\phi(t,y_0)$ for $x_0,y_0\in K$ it holds that, for all $t\geq 0$,
  \begin{align}
    \label{eq:7}
    |x(t)-y(t)|&\leq e^{ct}|x_0-y_0|, \quad \text{and}\\
  \label{eq:8}
  |f(x(t))|&\leq e^{ct}|f(x_0)|.
\end{align}
\end{prop}

\begin{proof}
A proof for condition \eqref{eq:7} is found in \cite{Sontag:2010fk} where it is assumed that $c<0$ although the proof holds without this assumption.

 Considering \eqref{eq:8}, let $V(x)\triangleq |f(x)|$. $V(x(t))$ is then absolutely continuous as a function of $t$ and therefore
    \begin{align}
      \label{eq:204}
      \dot{V}(x(t))&\triangleq \lim_{h\to 0^+}\frac{V(x(t+h))-V(x(t))}{h}%
    \end{align}
for almost all $t$. Furthermore,
\begin{align}
  \label{eq:208}
&\lim_{h\to 0^+}\bigg|\frac{|f(x(t+h))|-|f(x)+h\dot{f}(x)|}{h}\bigg|\\
&\leq \lim_{h\to 0^+}\left|\frac{f(x(t+h))-f(x)}{h}-\dot{f}(x)\right|\\
\label{eq:208-2}&=0
\end{align}
where we use the definition of $\dot{f}(x)$ and the fact $\big||x|-|y|\big|\leq |x-y|$. Since also $\dot{f}(x)=J(x)f(x)$, we combine \eqref{eq:204}--\eqref{eq:208-2} and obtain
\begin{align}
  \label{eq:206}
  \dot{V}(x(t))&=\lim_{h\to 0^+}\frac{|f(x)+hJ(x)f(x)|-|f(x)|}{h}\\
&\leq\lim_{h\to 0^+} \frac{\Vert I+hJ(x)\Vert \cdot |f(x)|-|f(x)|}{h}\\
&=\lim_{h\to 0^+} \frac{\Vert I+hJ(x)\Vert-1}{h}|f(x)|\\
&=\mu(J(x))V(x).
\end{align}
By hypothesis, we then have $\dot{V}(x)\leq c V(x)$, and \eqref{eq:8} follows by integration.
\end{proof}
  It is also possible to obtain \eqref{eq:8} using Coppel's inequality; see, \emph{e.g.}, \cite[Section 2.5, Theorem 3]{Vidyasagar:2002ly} for a statement and proof of the inequality.

\begin{definition}
The system \eqref{eq:1} is \emph{infinitesimally contracting} on $K$ with respect to the norm $|\wc |$ if \eqref{eq:15} holds for some $c<0$.
\end{definition}
 If the system is infinitesimally contracting, then $|f(x)|$ decays to zero at an exponential rate by \eqref{eq:8}, and therefore each trajectory converges to a finite equilibrium. Since \eqref{eq:7} precludes the existence of more than one equilibrium, we conclude that there exists a unique equilibrium, it is asymptotically stable, and the domain of attraction includes $K$.  Moreover, Proposition \ref{thm:contract} provides two possible Lyapunov functions defined in terms of the norm $|\wc|$. Namely, if $x^*$ is the unique equilibrium, then $V(x)=|x-x^*|$ and $V(x)=|f(x)|$ are both Lyapunov functions for \eqref{eq:1}.

For the $\ell_1$ norm with induced matrix norm $\|A\|_1=\max_{j}\sum_{i}|A_{ij}|$, the induced matrix measure is given by
\begin{align}
  \label{eq:56}
  \textstyle \mu_1(A)=\max_{j}\left(A_{jj}+\sum_{i\neq j}|A_{ij}|\right)
\end{align}
for any $A\in\mathbb{R}^{n\times n}$. Likewise, for the $\ell_\infty$ norm with induced matrix norm is $\|A\|_\infty=\max_{i}\sum_{j}|A_{ij}|$, the induced matrix measure is given by
\begin{align}
  \label{eq:58}
  \textstyle   \mu_\infty(A)=\max_{i}\left(A_{ii}+\sum_{j\neq i}|A_{ij}|\right).
\end{align}
See, \emph{e.g.}, \cite[Section II.8, Theorem 24]{Desoer:2008bh}, for a derivation of the induced matrix measures for common vector norms. The matrix measures given in \eqref{eq:56} and \eqref{eq:58} provide the connection to sum-separable and max-separable Lyapunov functions which are the focus of this paper. 

We are particularly interested in vector norms and matrix measures that arise from a scaling of another norm. Let $|\wc|_{*}$ be some particular vector norm and let $\mu_*(\cdot)$ be its induced matrix measure. If $P\in\mathbb{R}^{n\times n}$ is nonsingular, then we define a new vector norm by $|x|_{*,P}:=|Px|_{*}$ for which the induced matrix measure satisfies
\begin{align}
  \label{eq:11}
 \mu_{*,P}(A)=\mu_*(PAP^{-1}).  
\end{align}
When $P=\text{diag}(v)$ for some $v\in\mathbb{R}^n$, we use the notation $|\wc|_{*,v}$ and $\mu_{*,v}(\cdot)$ instead, where $\text{diag}(v)$ denotes the $n\times n$ matrix with $v$ on the diagonal and zeros elsewhere.

For some classes of systems, it is only possible to establish $\mu(J(x))\leq 0$ for all $x\in K$. In this case, \eqref{eq:15}--\eqref{eq:7} implies a nonexpansion property. Furthermore, it may still be possible to demonstrate asymptotic stability of an equilibrium using contraction theory.

\begin{thm}
\label{thm:ne}
  Let $K\subseteq \Domain$ be convex and forward-invariant and let $x^*\in K$ be an equilibrium of \eqref{eq:1}. If
  \begin{align}
    \label{eq:21}
    \mu(J(x))&\leq 0 \ \forall x\in K\qquad \text{ and} \qquad\mu(J(x^*))<0
  \end{align}
then $x^*$ is asymptotically stable, the domain of attraction includes $K$, and $V(x)=|x-x^*|$ and $V(x)=|f(x)|$ are Lyapunov functions.

\end{thm}

Theorem \ref{thm:ne} is closely related to existing results in the literature, although we believe the generality provided by Theorem \ref{thm:ne} and the further generalization to periodically time-varying systems in the appendix, for which we prove convergence to a unique periodic trajectory, is novel. In particular, \cite[Lemma 6]{Lovisari:2014yq} provides a similar result for $\mu(\cdot)$ restricted to the matrix measure induced by the $\ell_1$ norm under the assumption that \eqref{eq:1} is monotone. A similar technique is applied to periodic trajectories of a class of monotone flow networks in \cite[Proposition 2]{Lovisari:2014qv}, but a general formulation is not presented. 

We conclude with a final technical result that will be useful for constructing Lyapunov functions as the limit of a sequence of contraction metrics.

\begin{prop}
\label{prop:limit}
  Let $K\subseteq \Domain$ be forward-invariant and let $x^*\in K$ be an asymptotically stable equilibrium of \eqref{eq:1} for which the domain of attraction includes $K$.  Suppose there exists a sequence of Lyapunov functions $V^i:K\to\mathbb{R}_{\geq 0}$ for \eqref{eq:1} on $K$ that converges locally uniformly to $V(x):=\lim_{i\to\infty}V^i(x)$. If $V(x)$ is radially unbounded and $V(x)=0$ if and only if $x=x^*$, then $V(x)$ is also a Lyapunov function for \eqref{eq:1}.

\end{prop}

\begin{proof}
Consider some $x_0\in K$.  Asymptotic stability of $x^*$ implies there exists a bounded set $\Omega$ for which $\phi(t,x_0)\in\Omega\subseteq K$ for all $t\geq 0$. For $i=1,\ldots,n$, we have $V^i(\phi(t,x_0))$ is nonincreasing in $t$ and $\lim_{t\to\infty}V^i(\phi(t,x_0))=0$. Local uniform convergence establishes $V(x)$ is continuous, $V(\phi(t,x_0))$ is nonincreasing in $t$, and $\lim_{t\to\infty}V(\phi(t,x_0))=0$. Under the additional hypotheses of the proposition, we have that $V(x)$ is therefore a Lyapunov function.%
\end{proof}

Note that a sequence $V^i(x)$ arising from a sequence of weighted contraction metrics, \emph{i.e.}, $V^i(x)=|P_i(x-x^*)|$ or $V^i(x)=|P_if(x)|$ for $P_i$ converging to some nonsingular $P$, satisfies the conditions of Proposition \ref{prop:limit}.

\section{Contractive Monotone Systems}
\label{sec:contr-monot-syst}
In the remainder of this paper, we assume \eqref{eq:1} is monotone.  For $A$ Metzler, since $A_{ij}\geq 0$ for all $i\neq j$, %
\begin{align}
  \label{eq:10}
  \mu_1(A)&\textstyle =\max_{j=1,\ldots,n}\sum_{i=1}^nA_{ij},\\
  \mu_\infty(A)&\textstyle =\max_{i=1,\ldots,n}\sum_{j=1}^nA_{ij},%
\end{align}
that is, $\mu_1(A)$ is the largest column sum of $A$ and $\mu_{\infty}(A)$ is the largest row sum of $A$. The following proposition is easily verified from the identity \eqref{eq:11}.

\begin{prop}
\label{prop:metz}
For $A$ Metzler and $v\in\mathbb{R}^n$ with $v>0$, 
\begin{align}
  \label{eq:12}
 \mu_{1,v}(A)<c \quad \text{if and only if} \quad v^TA<cv^T.  
\end{align}
Likewise, for $w\in\mathbb{R}^n$ with $w>0$, 
\begin{align}
  \label{eq:13}
\mu_{\infty,w^\dagger}(A)<c\quad \text{if and only if} \quad A{w}<c{w}
\end{align}
where $w^\dagger:=(1/w_1,1/w_2,\ldots,1/w_n)$.
\end{prop}
Propositions \ref{prop:mono} and \ref{prop:metz} lead to the following theorems.
\begin{thm}
\label{thm:main1}
Let \eqref{eq:1} be a monotone system. If there exists $v>0$, $c<0$, and convex, forward invariant $K\subseteq \Domain$  such that
  \begin{align}
    \label{eq:16}
    v^TJ(x)\leq c\mathbf{1}^T \quad \forall x\in K,
  \end{align}
then there exists an asymptotically stable equilibrium $x^*\in K$ and the domain of attraction includes $K$.  %
Furthermore, either of the following are Lyapunov functions on $K$:
\begin{align}
  \label{eq:17}
  V(x)&= \sum_{i=1}^nv_i|x_i-x_i^*|,\\
  \label{eq:17-2}
V(x)&= \sum_{i=1}^nv_i|f_i(x)|.
\end{align}
\end{thm}

\begin{proof}
  Suppose \eqref{eq:16} holds. There exists some $\tilde{c}<0$ such that  $v^TJ(x)\leq \tilde{c}v$ for all $x\in K$, in particular, we take $\tilde{c}\in[c/|v|_\infty,0)$. From \eqref{eq:12}, it follows that $\mu_{1,v}(J(x))\leq \tilde{c}$. The theorem follows from Proposition \ref{thm:contract}.
\end{proof}

\begin{thm}
\label{thm:main2}
  Let \eqref{eq:1} be a monotone system. If there exists $w>0$, $c<0$, and convex, forward invariant $K\subseteq\Domain$ such that
  \begin{align}
    \label{eq:18}
    J(x)w\leq c\mathbf{1}\quad \forall x\in K,
  \end{align}
then there exists an asymptotically stable equilibrium $x^*\in K$ and the domain of attraction includes $K$.  %
Furthermore, either of the following are Lyapunov functions on $K$:
\begin{align}
  \label{eq:19}
  V(x)&=\max_{i=1,\ldots,n}\left\{\frac{1}{w_i}|x_i-x_i^*|\right\},\\
  \label{eq:20}
  V(x)&=\max_{i=1,\ldots,n}\left\{\frac{1}{w_i}|f_i(x)|\right\}.
\end{align}
\end{thm}

Note that Theorem \ref{thm:main1} and Theorem \ref{thm:main2} lead to global asymptotic stability when $K=\Domain$. Moreover, Theorems \ref{thm:main1} and \ref{thm:main2} can be considered nonlinear extensions of known stability verification results for linear monotone (that is, positive) systems. In particular, conditions \eqref{eq:16} and \eqref{eq:18} recover the stability conditions (1.2) and (1.3) of \cite[Proposition 1]{Rantzer:2012fj} when $J(x)$ is replaced with the static matrix $A$ as detailed in Example \ref{ex:1} below.
\begin{definition}
  The system \eqref{eq:1} is \emph{contractive monotone} if it is monotone and infinitesimally contracting.
\end{definition}
The hypotheses of Theorem \ref{thm:main1} and Theorem \ref{thm:main2} imply that the system is contractive monotone.  We now specialize Theorem \ref{thm:ne} to monotone systems.%

\begin{corollary}
\label{thm:ne1}
    Let \eqref{eq:1} be a monotone system with equilibrium $x^*$. If there exists $v>0$ and convex, forward invariant $K\subseteq \Domain $ such that
    \begin{align}
      \label{eq:9}
          v^TJ(x)&\leq 0\ \forall x\in K\quad \text{and}\quad      v^TJ(x^*)<0,
    \end{align}
then $x^*$ is asymptotically stable, the domain of attraction includes $K$, and \eqref{eq:17}--\eqref{eq:17-2} are Lyapunov functions.
\end{corollary}

\begin{corollary}
\label{thm:ne2}
    Let \eqref{eq:1} be a monotone system with equilibrium $x^*$. If there exists $w>0$ and convex, forward invariant $K\subseteq \Domain $ such that
  \begin{align}
    \label{eq:6}
          J(x)w&\leq 0\  \forall x\in K\quad \text{and}\quad J(x^*)w<0,
  \end{align}
then $x^*$ is asymptotically stable,  the domain of attraction includes $K$, and \eqref{eq:19}--\eqref{eq:20} are Lyapunov functions.
\end{corollary}

 \section{Examples}
\label{sec:examples}

We now present several examples. First, we recover a well-known condition for stability of monotone linear systems, also called positive linear systems.
 \begin{example}[Linear systems]
\label{ex:1}
Consider $\dot{x}=Ax$ for $A$ Metzler. Theorems \ref{thm:main1} and \ref{thm:main2} imply that if one of the following conditions holds, then the system is globally asymptotically stable:
\begin{align}
  \label{eq:53}
  &\text{There exists $v>0$ such that $v^TA<0$, \quad or}\\
  \label{eq:53-2}&\text{There exists $w>0$ such that $Aw<0$}.
\end{align}
If \eqref{eq:53} holds then $V(x)=\sum_{i=1}^nv_i|x_i|$ and $V(x)=\sum_{i=1}^n v_i|(Ax)_i|$ are Lyapunov functions, and if \eqref{eq:53-2} holds then $V(x)=\max_i\{|x_i|/w_i\}$ and $V(x)=\max_i\{|(Ax)_i|/w_i\}$ are Lyapunov functions where $(Ax)_i$ denotes the $i$th element of $Ax$.
\end{example}
In fact, it is well known that $A$ is Hurwitz if and only if either (and therefore, both) of the two conditions \eqref{eq:53} and \eqref{eq:53-2} hold, as established in, \emph{e.g.}, \cite [Proposition 1]{Rantzer:2012fj}, and the corresponding state separable Lyapunov functions of Example \ref{eq:1} are also derived in \cite{Rantzer:2012fj}.

The following example is inspired by \cite [Example 3]{Dirr:2015rt}.
\begin{example}[Comparison system]
  Consider the system 
  \begin{align}
    \label{eq:23}
    \dot{x}_1&=-x_1+x_1x_2\\
    \label{eq:23-2}\dot{x}_2&=-2x_2-x_2^2+\gamma(x_1)^2
  \end{align}
evolving on $\Domain=\mathbb{R}_{\geq 0}^2$ where $\gamma:\mathbb{R}_{\geq 0}\to\mathbb{R}_{\geq 0}$ is strictly increasing  and satisfies $\gamma(0)=0$, $\bar{\gamma}:=\lim_{\sigma\to\infty}\gamma(\sigma)<1$, and $  \gamma'(\sigma)\leq \frac{1}{(1+\sigma)^{2}}$.
Consider the change of coordinates $(\eta_1,\eta_2)=(\log(1+x_1),x_2)$ so that 
\begin{align}
  \label{eq:27}
  \dot{\eta}_1&=\frac{1}{1+x_1}(-x_1+x_1x_2)
\end{align}
where we substitute $(x_1,x_2)=(e^{\eta_1}-1,\eta_2)$. Then
\begin{align}
  \label{eq:28}
  \dot{\eta}_1\leq -\beta(e^{\eta_1}-1)+\eta_2
\end{align}
where $\beta(\sigma)={\sigma}/{(1+\sigma)}$. Introduce the comparison system 
\begin{align}
  \label{eq:29}
  \dot{\xi}_1&=-\beta(e^{\xi_1}-1)+\xi_2\\
  \label{eq:29-2}  \dot{\xi}_2&=-2\xi_2-\xi_2^2+\gamma(e^{\xi_1}-1)^2
\end{align}
evolving on $\mathbb{R}_{\geq 0}^2$. The comparison principle (see, \emph{e.g.}, \cite{Dirr:2015rt}) ensures that asymptotic stability of the origin for the comparison system \eqref{eq:29}--\eqref{eq:29-2} implies asymptotic stability of the origin of the $(\eta_1,\eta_2)$ system, which in turn establishes asymptotic stability of the origin for \eqref{eq:23}--\eqref{eq:23-2}. The Jacobian of \eqref{eq:29}--\eqref{eq:29-2} is given by
\begin{align}
  \label{eq:30}
  J(\xi)=
  \begin{pmatrix}
    -e^{\xi_1}\beta'(e^{\xi_1}-1)&1\\
2e^{\xi_1}\gamma(e^{\xi_1}-1)\gamma'(e^{\xi_1}-1)&-2-2\xi_2
  \end{pmatrix}
\end{align}
where $\beta'(\sigma)=\frac{1}{(1+\sigma)^2}$. Let $v=(2\bar{\gamma}+\epsilon,1)$ where $\epsilon$ is chosen small enough so that $c_1:=(2\bar{\gamma}+\epsilon-2)<0$. It follows that
\begin{align}
  \label{eq:14}
  v^TJ(\xi)\leq (-\epsilon e^{-\xi_1},c_1)\leq 0\quad \forall \xi
\end{align}
and $v^TJ(0)<0$. Applying Corollary \ref{thm:ne1}, the origin of \eqref{eq:23}--\eqref{eq:23-2} and \eqref{eq:29}--\eqref{eq:29-2} is globally asymptotically stable. Furthermore, we have the following state and flow sum-separable Lyapunov functions for the comparison system \eqref{eq:29}--\eqref{eq:29-2}:
\begin{align}
  \label{eq:50}
 V(\xi)&=(2\bar{\gamma}+\epsilon)\xi_1+\xi_2  \\
V(\xi)&=(2\bar{\gamma}+\epsilon)|\dot{\xi}_1|+|\dot{\xi_2}|.
\end{align}
Above, we understand $\dot{\xi}_1$ and $\dot{\xi}_2$ to be shorthand for the equalities expressed in \eqref{eq:29}--\eqref{eq:29-2}.
\end{example}

\begin{example}[Multiagent system]
\label{ex:multiagent}
Consider the following system evolving on $\Domain=\mathbb{R}^3$:
\begin{align}
  \label{eq:31}
  \dot{x}_1&=-\alpha_1(x_1)+\rho_1(x_3-x_1)\\
\dot{x}_2&=\rho_2(x_1-x_2)+\rho_3(x_3-x_2)\\
  \label{eq:31-3}\dot{x}_3&=\rho_4(x_2-x_3)
\end{align}
where we assume $\alpha_1:\mathbb{R}\to\mathbb{R}$ is strictly increasing and satisfies $\alpha(0)=0$  and $\alpha_1'(\sigma)\geq \ul{c}_0$ for some $\ul{c}_0>0$ for all $\sigma$, and each $\rho_i:\mathbb{R}\to\mathbb{R}$ is strictly increasing and satisfies $\rho_i(0)=0$. Furthermore, for $i=1,3$, $\rho'_i(\sigma)\leq \ol{c}_i$ for some $\ol{c}_i>0$ for all $\sigma$, and for $i=2,4$, $\rho'_i(\sigma)\geq \ul{c}_i$ for some $\ol{c}_i>0$ for all $\sigma$. 

For example, $x_1$, $x_2$, and $x_3$ may be the position of three vehicles, for which the dynamics  \eqref{eq:31}--\eqref{eq:31-3}  are a rendezvous protocol whereby agent 1 moves towards agent 3 at a rate dependent on the distance $x_3-x_1$ as determined by $\rho_1$, \emph{etc.} Additionally, agent 1 navigates towards the origin according to $-\alpha_1(x_1)$. Computing the Jacobian, we obtain
\begin{align}
  \label{eq:32}
\nonumber&  J(x)=\\
  &\begin{pmatrix}
    -\alpha'(x_1)-\rho_1'(z_{31})&0&\rho_1'(z_{31})\\
\rho_2'(z_{12})&-\rho_2'(z_{12})-\rho_3'(z_{32})&\rho_3'(z_{32})\\
0&\rho'_4(z_{23})&-\rho'_4(z_{23})
  \end{pmatrix}
\end{align}
where $z_{ij}:= x_i-x_j$. Let $w=(1,1+\epsilon_1,1+\epsilon_1+\epsilon_2)^T$
where $\epsilon_1>0$ and $\epsilon_2>0$ are chosen to satisfy
\begin{align}
  \label{eq:33}
\ul{c}_0&>  (\epsilon_1+\epsilon_2)\ol{c}_1\quad \text{and}\quad \epsilon_1\ul{c}_2>\epsilon_2\ol{c}_3.%
\end{align}
We then have $  J(x)w\leq c\mathbf{1}$ for all $x$ 
for $c=\max\{(\epsilon_1+\epsilon_2) \ol{c}_1-\ul{c}_0 ,\epsilon_2\ol{c}_3-\epsilon_1\ul{c}_2,-\epsilon_2\ul{c}_4\}<0$. Thus, the origin of \eqref{eq:31}--\eqref{eq:31-3} is globally asymptotically stable by Theorem \ref{thm:main2}. Furthermore,
\begin{align}
  \label{eq:37}
  V(x)&=\max\{|x_1|,(1+\epsilon_1)^{-1}|x_2|,(1+\epsilon_1+\epsilon_2) ^{-1}|x_3|\},\\
  \label{eq:37-2}   V(x)&=\max\{|\dot{x}_1|,(1+\epsilon_1) ^{-1}|\dot{x}_2|,(1+\epsilon_1+\epsilon_2) ^{-1}|\dot{x}_3|\}
\end{align}
are state and flow max-seperable Lyapunov functions where we interpret $\dot{x}_i$ as shorthand for the equalities expressed in \eqref{eq:31}--\eqref{eq:31-3}. %
Since we can take $\epsilon_1$ and $\epsilon_2$ arbitrarily small satisfying \eqref{eq:33},  using Proposition \ref{prop:limit} we have also the following choices for Lyapunov functions:
\begin{align}
  \label{eq:49}
  V(x)&=\max\{|x_1|,|x_2|,|x_3|\},\\
  \label{eq:49-2}   V(x)&=\max\{|\dot{x}_1|,|\dot{x}_2|,|\dot{x}_3|\} .
\end{align}

The flow max-separable Lyapunov functions \eqref{eq:37-2} and \eqref{eq:49-2} are particularly useful for multiagent vehicular networks where it often easier to measure each agent's velocity rather than absolute position.
\end{example}
In Example \ref{ex:multiagent}, choosing $w=\mathbf{1}$, we have $J(x)w\leq 0$, however this is not enough to establish asymptotic stability using Theorem \ref{thm:main2}. Informally, choosing $w$ as in the example
distributes the extra negativity of $-\alpha'(x_1)$ among the columns of $J(x)$. Nonetheless, Proposition \ref{prop:limit} implies that choosing $w=\mathbf{1}$ indeed leads to a valid Lyapunov function.%

The above example generalizes to systems with many agents interacting via arbitrary directed graphs, as does the principle of distributing extra negativity  along diagonal entries of the Jacobian as discussed in Section \ref{sec:disc-comp-exist}.

\begin{example}[Traffic flow]
\label{ex:traffic}
A model of traffic flow along a freeway with no onramps is obtained by spatially partitioning the freeway into $n$ segments such that traffic flows from segment $i$ to $i+1$, $x_i\in[0,\bar{x}_i]$ is the density of vehicles occupying link $i$, and $\bar{x}_i$ is the capacity of link $i$. A fraction $\beta_i\in(0,1]$ of the flow out of link $i$ enters link $i+1$. The remaining $1-\beta_i$ fraction is assumed to exit the network via, \emph{e.g.}, unmodeled offramps. Associated with each link is a continuously differentiable \emph{demand} function $D_i:[0,\bar{x}_i]\to\mathbb{R}_{\geq 0}$ that is strictly increasing and satisfies $D_i(0)=0$, and a continuously differentiable \emph{supply} function $S_i:[0,\bar{x}_i]\to\mathbb{R}_{\geq 0}$ that is strictly decreasing and satisfies $S_i(\bar{x}_i)=0$. Flow from segment to segment is restricted by upstream demand and downstream supply, and the change in density of a link is governed by mass conservation:
\begin{align}
  \label{eq:38}
\dot{x}_1&= \min\{\delta_1,S_1(x_1)\}-\frac{1}{\beta_1}g_{1}(x_{1},x_{2})\\
  \dot{x}_i&= g_{i-1}(x_{i-1},x_i)-\frac{1}{\beta_i}g_{i}(x_{i},x_{i+1}), \quad i=2,\ldots,n-1\\
  \label{eq:38-3}\dot{x}_n&=g_{n-1}(x_{n-1},x_n)- D_n(x_n)
\end{align}
for some $\delta_1>0$ where, for $i=1,\ldots,n-1$,
\begin{align}
  \label{eq:39}
g_{i}(x_{i},x_{i+1})=\min\{\beta_iD_i(x_i),S_{i+1}(x_{i+1})\}.
\end{align}
Let $\delta_i\triangleq \delta_1\prod_{j=1}^{i-1}\beta_j$ for $i=2,\ldots, n$. If $d^{-1}_i(\delta_i)<s^{-1}_i(\delta_i)$ for all $i$, then $\delta_1$ is said to be \emph{feasible} and $x^*_i:=d^{-1}_i(\delta_i)$ constitutes the unique equilibrium. 

\begin{figure*}
\begin{align}
  \label{eq:36}
  J(x)=
  \begin{pmatrix}
    \partial_1g_0-\frac{1}{\beta_1}\partial_1g_1&-\frac{1}{\beta_1}\partial_2g_1&0&0&\cdots&0\\
\partial_1 g_1&\partial_2g_1-\frac{1}{\beta_2}\partial_2g_2&-\frac{1}{\beta_2}\partial_3 g_2&0 &\cdots&0\\
0&\partial_2g_2&\partial_3g_2-\frac{1}{\beta_3}\partial_3g_3&-\frac{1}{\beta_3}\partial_4g_3&&0\\
\vdots&&&&\ddots&\vdots\\
0&0&\cdots&0&\partial_{n-1}g_{n-1}&\partial_{n}g_{n-1}-\partial_nD_n(x_n)
  \end{pmatrix}
\end{align}
\hrule
\end{figure*}

Let $\partial_i$ denote differentiation with respect to the $i$th component of $x$, that is, $\partial_ig(x):=\frac{\partial g}{\partial x_i}(x)$ for a function $g(x)$. The dynamics \eqref{eq:38}--\eqref{eq:38-3} define a system $\dot{x}=f(x)$ for which $f$ is continuous but only piecewise differentiable. Nonetheless, the results developed above apply for this case, and, in the sequel, we interpret statements involving derivatives to hold wherever the derivative exists.

Notice that $\partial_{i}g_i(x_i,x_{i+1})\geq 0$ and $\partial_{i+1}g_i(x_i,x_{i+1})\leq 0$. Define $g_0(x_1):=\min\{\delta_1,S_1(x_1)\} $. 
The Jacobian, where it exists, is given by \eqref{eq:36} on the following page. Let
\begin{align}
  \label{eq:40}
  \tilde{v}=\begin{pmatrix}1,\beta_1^{-1},(\beta_1\beta_2)^{-1},\ldots,(\beta_1\beta_2\cdots\beta_{n-1})^{-1}\end{pmatrix}^T.
\end{align}
Then $\tilde{v}^TJ(x)\leq 0$ for all $x$. Moreover, there exists $\epsilon=(\epsilon_1,\epsilon_2,\ldots,\epsilon_{n-1},0)$ with $\epsilon_{i}>\epsilon_{i+1}$ for each $i$ such that $v:=\tilde{v}+\epsilon$ satisfies
\begin{align}
  \label{eq:41}
  v^TJ(x)&\leq 0\quad \forall x\\
  \label{eq:41-2} v^TJ(x^*)&<0.
\end{align}
Such a vector $\epsilon$ is constructed using a technique similar to that used in Example \ref{ex:multiagent}. In particular, the sum of the $n$th column of $\text{diag}(\tilde{v})J(x)$ is strictly negative because $-\partial_nD_n(x_n)<0$, and this excess negativity is used to construct $v$ such that \eqref{eq:41}--\eqref{eq:41-2} holds. A particular choice of $\epsilon$ such that \eqref{eq:41}--\eqref{eq:41-2} holds depends on bounds on the derivative of the demand functions $D_i$, but it is possible to choose $\epsilon$ arbitrarily small. Corollary~\ref{thm:ne1} establishes asymptotic stability, and Proposition \ref{prop:limit} gives the following sum-separable Lyapunov functions:
\begin{align}
  \label{eq:42}
  V(x)&=\sum_{i=1}^n \left(|x_i|\prod_{j=1}^{i-1}\beta_j\right),\\
  \label{eq:42-2}  V(x)&=\sum_{i=1}^n \left(|\dot{x}_i|\prod_{j=1}^{i-1}\beta_j\right),
\end{align}
where we interpret $\dot{x}_i$ according to \eqref{eq:38}--\eqref{eq:38-3}.

In traffic networks, it is often easier to measure traffic flow rather than traffic density. Thus, \eqref{eq:42-2} is a practical Lyapunov function indicating that the (weighted) total absolute net flow throughout the network decreases over time.
\end{example}
 In \cite{coogan2015compartmental}, a result similar to that of Example \ref{ex:traffic} is derived for possibly infeasible input flow and  traffic flow network topologies where merging junctions with multiple incoming links are allowed. The proof considers a flow sum-separable Lyapunov function similar to \eqref{eq:42-2} and appeals to LaSalle's invariance principle rather than Proposition \ref{prop:limit}.

\section{Discussion}
\label{sec:disc-comp-exist}

We first highlight the connection of the above results to small-gain conditions for interconnections of nonlinear systems. Consider $N$ interconnected systems with dynamics $\dot{x}_i=f_i(x_1,\ldots,x_N)$ for $x_i\in\mathbb{R}^{n_i}$ and suppose they satisfy an input-to-state stability (ISS) condition \cite{Sontag:1989fk} whereby there exists ISS Lyapunov functions $V_i$ \cite{Sontag:1995qf} satisfying
\begin{align}
  \label{eq:57}
  \frac{\partial V_i}{\partial x_i}(x_i)f_i(x)\leq -\alpha_i(V_i(x_i))+\sum_{i\neq j}\gamma_{ij}(V_j(x_j))
\end{align}
where each $\alpha_i$ and $\gamma_{ij}$ is a $\Kinf$ function. We obtain a monotone comparison system 
\begin{align}
  \label{eq:51}
\dot{\nu}=g(\nu), \qquad g_i(\nu)=-\alpha_i(\nu_i)+\sum_{j\neq i}\gamma_{ij}(\nu_j)
\end{align}
evolving on $\mathbb{R}^n_{\geq 0}$ for which asymptotic stability of the origin implies asymptotic stability of the original system \cite{Ruffer:2010tw}. For $N=2$, it is noted in \cite{Jiang:1996dw} that if $\gamma_{12}(\sigma)=\kappa_1\alpha_2 (\sigma)$ and $\gamma_{21}(\sigma)=\kappa_2\alpha_1 (\sigma)$ for $\kappa_1>0$, $\kappa_2>0$ such that $\kappa_1\kappa_2<1$, then $v_1V_1(x_1)+v_2V(x_2)$ is a Lyapunov function for the original system for any $v=(v_1\ v_2)^T>0$ satisfying $v_1\kappa_1<v_2<v_1/\kappa_2$. Indeed, for such a choice, we see that $v^T\frac{\partial g}{\partial \nu}(\nu)\leq 0$, and if additionally $\alpha_i'(0)>0$ for $i=1,2$, then $v^T\frac{\partial g}{\partial \nu}(0)<0$ so that Corollary \ref{thm:ne1} provides a contraction theoretic interpretation of this result.

 Alternatively, in \cite{Ruffer:2010tw, Dashkovskiy:2010zh}, it is shown that if there exists a function $\rho:\mathbb{R}_{\geq 0}\to\mathbb{R}_{\geq 0}^n$ with each component $\rho_i$ belonging to class\footnote{A continuous function $\alpha:\mathbb{R}_{\geq 0}\to\mathbb{R}_{\geq 0}$ is of class $\Kinf$ if it is strictly increasing, $\alpha(0)=0$, and $\lim_{r\to\infty}\alpha(r)=\infty$.} $\mathcal{K}_\infty$ such that $g(\rho(r))< 0$ for all $r>0$, then the origin is asymptotically stable and $V(\nu):=\max_i\{\rho_i^{-1}(\nu_i)\}$ is a Lyapunov function. If condition \eqref{eq:6} of Corollary \ref{thm:ne2} holds for the comparison system for some $w$, we may choose $\rho(r)=rw$. Indeed, we have
\begin{align}
  \label{eq:55}
  g(rw)=\int_0^1\frac{\partial g}{\partial \nu}(\sigma rw) rw\ d \sigma <0 \quad \forall r>0.
\end{align}
 For this case, $V(\nu)=\max_i\{\rho_i^{-1}(\nu_i)\}=\max_i\{\nu_i/w_i\}$, recovering \eqref{eq:19}.

Next, we comment on the relationship between Theorem \ref{thm:ne} as well as Proposition \ref{prop:limit} and a generalization of contraction theory recently developed in \cite{Sontag:2014eu, Margaliot:2015wd} where exponential contraction between any two trajectories is required only after an arbitrarily small amount of time, an arbitrarily small overshoot, or both. In \cite[Corollary 1]{Margaliot:2015wd}, it is shown that if a system is contractive with respect to a sequence of norms convergent to some norm, then the system is generalized contracting with respect to that norm, a result analogous to Proposition \ref{prop:limit}. In \cite{Margaliot:2015wd}, conditions on the sign structure of the Jacobian are obtained that ensure the existence of such a sequence of weighted $\ell_1$ or $\ell_\infty$ norms. These conditions are a generalization of the technique in Example \ref{ex:multiagent}
and Example \ref{ex:traffic} 
above where small $\epsilon$ is used to distribute excess negativity.  

Furthermore, it is shown in \cite{Margaliot:2012hc, Margaliot:2014qv} that a ribosome flow model for gene translation is monotone and nonexpansive with respect to a weighted $\ell_1$ norm, and additionally is contracting on a subset of its domain. %
Entrainment of solutions is proved by first showing that all trajectories reach the region of exponential contraction. Theorem \ref{thm:appendix} in the appendix provides a different approach for studying entrainment by observing that the distance to the periodic trajectory strictly decreases in each period due to a neighborhood of contraction along the periodic trajectory. Theorem \ref{thm:ne} provides an analogous result for stability analysis of an equilibrium.%

Finally, we note that Metzler matrices with nonpositive column sums have also been called \emph{compartmental} \cite{Jacquez:1993uq}. It has been shown that if the Jacobian matrix is compartmental for all $x$, then $V(x)=|f(x)|$ is a decreasing function along trajectories of \eqref{eq:1} \cite{Jacquez:1993uq, Maeda:1978fk}; we recover this observation by considering the Lyapunov function implied by $\eqref{eq:8}$ with $c=0$ and $|\wc|$ taken to be the standard $\ell_1$ norm.

\section{Conclusions}
\label{sec:conclusions}
We have investigated monotone systems that are also contracting with respect to a weighted $\ell_1$ norm or $\ell_\infty$ norm. In the case of the $\ell_1$ (respectively, $\ell_\infty$) norm, we provided a condition on the weighted column (respectively, row) sums of the Jacobian matrix for ensuring contraction. These conditions lead to either sum-separable or max-separable Lyapunov functions. In particular, we introduce a class of separable Lyapunov functions that depend on the value of the vector field along trajectories of the system. These flow separable Lyapunov functions are especially relevant in applications where it is easier to measure the derivative of the system's state rather than measure the state directly.

Paralleling observations made in \cite{Rantzer:2012fj}, verifying \eqref{eq:16} and \eqref{eq:18} requires checking nonpositivity of a collection of $n$ functions. For polynomial or rational vector fields, this is done efficiently using sum-of-squares (SOS) techniques \cite{Parrilo:2000fk}. Future work will consider scalable verification and synthesis methods for contractive monotone systems using SOS techniques. %

\section{Acknowledgements}
The author thanks Murat Arcak for providing valuable feedback on an early draft of this paper.

\appendix
\section{Proof of Theorem \ref{thm:ne}}
To prove Theorem \ref{thm:ne}, we prove a more general result for which Theorem \ref{thm:ne} is a special case.   Consider
\begin{equation}
  \label{eq:101}
  \dot{x}=f(t,x)
\end{equation}
for $x\in \Domain$. We assume that $f(t,x)$ is periodic in $t$ with period $T$. We further assume $f(t,x)$ is differentiable in $x$ and that $f(t,x)$ and the Jacobian $J(t,x)\triangleq \frac{\partial f}{\partial x}(t,x)$ are continuous in $(t,x)$.

\begin{thm}
\label{thm:appendix}
Let $K\subseteq \Domain$ be convex and forward-invariant and suppose $\gamma(t)$ be a periodic trajectory of \eqref{eq:101} with period $T$. If 
\begin{align}
  \label{eq:102}
  \mu(J(t,x))\leq 0\quad \forall t\geq 0, \forall x\in K
\end{align}
and there exists $t^*$ such that
\begin{align}
  \label{eq:59}
  \mu(J(t^*,\gamma(t^*)))<0,
\end{align}
then for all trajectories $x(t)=\phi(t,x_0)$, $x_0\in K$ we have
\begin{align}
  \label{eq:104}
  \lim_{t\to \infty} |x(t)-\gamma(t)| = 0.
\end{align}
\end{thm}

\newcommand{\Ball}{\mathcal{B}}
\begin{proof}
  Without loss of generality, assume $t^*=0$. Condition \eqref{eq:59} and continuity of $J(t,x)$ imply there exists $\epsilon>0$, $c>0$, and $0<\tau\leq T$ such that
  \begin{align}
    \label{eq:60}
\mu(J(t,y))\leq -c   \qquad \forall t\in[0,\tau], \ \forall y\in B_\epsilon(\gamma(t))
  \end{align}
where $\Ball_\epsilon(y)=\{z:|z-y|\leq \epsilon\}$ denotes the closed ball of radius $\epsilon$ centered at $y\in\mathbb{R}^n$.   Define the mapping
\begin{align}
  \label{eq:108}
  P(\xi)=\phi(T,\xi)
\end{align}
and observe that $P^k(\xi)=\phi(kT,\xi)$. Let $\gamma^*=\gamma(0)$ and note that $\gamma^*$ is a fixed point of $P$. 
Consider a point $\xi\in\Ball_\epsilon(\gamma^*)$. Let $x(t)=\phi(t,\xi)$ and note that $|P(\xi)-\gamma^*|=|x(T)-\gamma(T)|$.
By \eqref{eq:102}, we have $ |x(T)-\gamma(T)|\leq |x(\tau)-\gamma(\tau)|$
and, by \eqref{eq:60}, we have \cite{Sontag:2010fk}
\begin{align}
  \label{eq:1017}
|x(\tau)-\gamma(\tau)|\leq e^{-c\tau}|x(0)-\gamma(0)|.  
\end{align}

Now consider $\xi\in K$ such that $|\xi-\gamma^*|> \epsilon$ and again let $x(t)=\phi(t,\xi)$. Let $\sigma(r)=r\xi+(1-r)\gamma^*$ be the parameterized line segment connecting $\gamma^*$ and $\xi$, and let $r_\epsilon$ be such that $|\sigma(r_\epsilon)-\gamma^*|=\epsilon$. Note that $|\xi-\gamma^*|=|\xi-\sigma(r_\epsilon)|+|\sigma(r_\epsilon)-\gamma^*|$. Let $s(t)=\phi(t,\sigma(r_\epsilon))$. By \eqref{eq:102}, we have $|x(T)-s(T)|\leq|x(0)-s(0)|= |\xi-\sigma(r_\epsilon)|=|\xi-\gamma^*|-\epsilon$. Furthermore, by the same argument as in the preceding case, we have $|s(T)-\gamma(T)|\leq e^{-c\tau}|\sigma(r_\epsilon)-\gamma^*|=e^{-c\tau}\epsilon$, and thus by the triangle inequality,
\begin{align}
  \label{eq:1012}
  |P(\xi)-\gamma^*&|\leq |x(T)-s(T)|+|s(T)-\gamma(T)|\\
&\leq |\xi-\gamma^*|-(1-e^{-c\tau})\epsilon\\
&=|\xi-\gamma^*|-\delta
\end{align}
where   $\delta:=(1-e^{-c\tau})\epsilon>0$. Then
\begin{align}
  \label{eq:105}
  |P(\xi)-\gamma^*|\leq 
  \begin{cases}
    |\xi-\gamma^*|-\delta&\text{if }|\xi-\gamma^*|>\epsilon\\
    e^{-c\tau}|\xi-\gamma^*|&\text{if }|\xi-\gamma^*|\leq \epsilon.
  \end{cases}
\end{align}
It follows that for all $\xi$, $|P^k(\xi)-\gamma^*|\leq \epsilon$ for some finite $k$ (in particular, for any $k\geq |\xi-\gamma^*|/\delta$). The theorem then follows from the second condition of \eqref{eq:105}.

\end{proof}

Theorem \ref{thm:ne} follows by taking $\gamma(t)\equiv x^*$ and arbitrary $T>0$, and we may take $\tau=T$ in the proof.%

\bibliographystyle{ieeetr}

\end{document}